\documentclass[conference]{IEEEtran}
\IEEEoverridecommandlockouts
\usepackage{cite}
\usepackage{amsmath,amssymb,amsfonts}
\usepackage{graphicx}
\usepackage{textcomp}
\usepackage{xcolor}
\usepackage{amsthm}
\usepackage{bm}
\usepackage{overpic}
\usepackage{float}
\usepackage{amsthm}
\usepackage{algorithm,algorithmic}
\usepackage[caption=false,font=footnotesize]{subfig}

\theoremstyle{plain}
\newtheorem{theorem}{Theorem}

\newtheorem{lemma}{Lemma}
\newtheorem{corollary}{Corollary}

\makeatletter
\newcommand\fs@ruled@notop{\def\@fs@cfont{\bfseries}\let\@fs@capt\floatc@ruled
  \def\@fs@pre{}%
  \def\@fs@post{\kern2pt\hrule\relax}%
  \def\@fs@mid{\kern2pt\hrule\kern2pt}%
  \let\@fs@iftopcapt\iftrue}
\renewcommand\fst@algorithm{\fs@ruled@notop}
\makeatother

\begin{document}

\title{Fundamentals of Energy-Efficient Hardware Configurations for Wireless Links with Sleep Modes
}

\author{Anders~Enqvist, {\"O}zlem~Tu\u{g}fe~Demir, Cicek~Cavdar, and~Emil~Bj{\"o}rnson%
\thanks{The preliminary version \cite{enqvist2024fundamentals} of this work was presented at the 2024 IEEE Vehicular Technology Conference (VTC2024-Spring) in Singapore.}
\thanks{This work was supported by the FFL18-0277 grant from the Swedish Foundation for Strategic Research and the Swedish Innovation Agency (Vinnova) through the SweWIN center (2023-00572).}
\thanks{A. Enqvist, C. Cavdar, and E. Bj{\"o}rnson are with the Department of Communication Systems, KTH Royal Institute of Technology, SE-100 44 Stockholm, Sweden. Email: enqv@kth.se, cavdar@kth.se, emilbjo@kth.se}%
\thanks{{\"O}. T. Demir is with the Department of Electrical and Electronics Engineering, Bilkent University, 06800, Ankara, Türkiye. Email: ozlemtugfedemir@bilkent.edu.tr}%
}

\maketitle

\begin{abstract}
In this paper, we examine the energy efficiency (EE) of a base station (BS) with multiple antennas. We use a state-of-the-art power consumption (PC) model that captures the passive and active parts of the transceiver circuitry, including the effects of radiated power, signal processing, and passive consumption. The paper treats the transmit power, bandwidth, and number of antennas as the optimization variables. We provide novel closed-form solutions for the optimal ratios of power per unit bandwidth and power per transmit antenna, and discover a new relationship in which the radiated power equals the total transceiver power at the EE-optimal operating point. A central finding is that the EE-optimal signal-to-noise ratio (SNR) collapses to a universal numerical constant of approximately 5.93\,dB, independent of channel and hardware parameters. We present an algorithm that jointly optimizes the three design variables to achieve maximum EE under practical constraints, and provide analytical insight into whether maximum power or maximum bandwidth is optimal and how many antennas a BS should utilize. We further extend the optimization framework to incorporate quality-of-service (QoS) requirements and three advanced sleep modes of varying depth: absolute sleep, deep sleep, and idle mode. We characterize the optimal hardware configuration for each mode and determine when the rush-to-sleep strategy, which transmits briefly at the EE-optimal active configuration and sleeps the rest of the time, is optimal. Incorporating wake-up transition delays, we reveal how latency constraints and sleep-mode-specific transition times jointly dictate the optimal sleep mode for data packets with absolute deadlines. Together, these results indicate that energy-efficient operation requires treating transmission and sleep as a single coupled optimization.

\end{abstract}

\begin{IEEEkeywords}
Energy efficiency, optimization, 6G, multiple antenna communications, sleep modes.
\end{IEEEkeywords}

\section{Introduction}

Mobile networks consumed approximately $290$\,TWh of electricity in 2023, with the radio access network (RAN) alone accounting for up to $87\%$ of mobile-network energy use \cite{kolta2024measuring}. With traffic continuing to grow \cite{ericssonmobilityreport2023}, driven by video demand and emerging applications such as immersive media, AI-assisted services, and integrated sensing, total network energy consumption (EC) is projected to rise further. Energy-efficient design is increasingly studied for these applications \cite{yu2026tree, zou2024energy}. The need for energy efficiency (EE), defined as the data rate divided by the related power consumption (PC), applies regardless of which applications dominate. The International Telecommunication Union has designated sustainability as a capability of IMT-2030 (6G) and established targets to maximize transmitted bits per Joule \cite{itu2023imt2030}; a notable shift from earlier generations whose targets focused on peak rates and spectral efficiency. While much of wireless evolution has traditionally focused on increasing capacity through technologies such as massive MIMO (multiple-input multiple-output) \cite{bjornson2017massive} and wider bandwidths in mmWave and sub-Terahertz bands \cite{Rappaport2019a}, recent efforts increasingly complement these advances with energy-saving mechanisms. For example, standardized advanced sleep modes \cite{3gpp.38.864}, which temporarily deactivate hardware components during low-traffic periods, can significantly reduce RAN EC. A broader overview of these and related energy-saving strategies is provided in \cite{lopez2022survey}.

\subsection{Prior Work and Motivations}

EE optimization was pioneered in \cite{Verdu1990a}, which identified a tradeoff between EE and rates. The paper \cite{auer2011much} was an early demonstration that the power needed to run a wireless system is much more than just the transmit power, motivating more realistic PC models. The papers \cite{Bjornson2016aabb,Zappone2023tradeoff} present a network model that considers the optimization of the area power consumption and area EE under rate constraints. In coordinated multi-cell downlink, \cite{venturino2014energy} jointly optimized scheduling and power control, showing that the operating points maximizing EE and spectral efficiency become increasingly aligned under stringent transmit-power constraints.

A line of work has focused on accurate PC models for BS hardware. The parameterized PC model in \cite{bjornson2015a} captures the fundamental behaviors of BSs and was later extended to include carrier aggregation in \cite{lopez2021energy}. In \cite{piovesan2022machine,piovesan2022power}, machine learning (ML) was employed to fit measured hardware data, demonstrating that an analytical linear PC model with appropriately chosen constants can accurately represent the PC in current BSs. The paper \cite{muneer2020handling} analyzed how power amplifier (PA) nonlinearities force a tradeoff between minimizing the spread of out-of-band emissions and optimizing the PC.

BSs are not always transmitting due to network traffic fluctuations, making it important to determine how to power down parts of the hardware during idle periods using advanced sleep modes. From a network point of view, \cite{andersson2016energy} showed that macro BSs can maximize their time spent in sleep modes by offloading traffic to small cells, effectively minimizing the total network PC in an area. Recent works at the link level minimize the BS PC by either adjusting time-slot activation for single-antenna rush-to-sleep strategies \cite{rottenberg2024information}, or numerically optimizing time, space, and power resources in single-band systems \cite{peschiera2025optimizing}. In parallel, data-driven and network-level approaches for sleep-mode management under uncertainty have been gaining traction. Deep reinforcement learning assisted by digital twins has been used to anticipate performance degradation before deactivating BS components \cite{masoudi2022digital}. In multi-cell massive MIMO deployments, multi-agent reinforcement learning has been used to jointly control antenna switching and multi-level sleep modes while managing inter-cell interference \cite{cai2024multi}. Beyond conventional cellular topologies, EE optimization has been extended to cell-free Open-RAN architectures through end-to-end orchestration of radio, fronthaul, and cloud processing resources \cite{demir2024cell}.

While prior works have examined EE optimization in increasingly complex networks, most studies focus on network-level objectives, heuristic scheduling, or numerical resource allocation. This paper instead returns to the physical fundamentals by considering a single communication link between a multi-antenna BS and a single-antenna user equipment (UE). The simplified setting enables an analytical characterization of the EE-optimal operating point: closed-form scaling laws, and explicit relations between transmit power, bandwidth, and the number of antennas. When seeking such an EE-optimal design, the adopted PC model heavily dictates the solution. When only accounting for the transmit power, the optimum is achieved as the rate approaches zero \cite{Verdu1990a}. By contrast, \cite{bjornson2018energy} studied the upper EE limit with a more detailed PC model and demonstrated that very different operating points may be reached. Our framework follows the latter view, jointly optimizing transmit power, bandwidth, and the number of antennas alongside the choice of sleep mode.

The importance of traffic-aware operation is supported by current network statistics. Downlink traffic is predominantly bursty and heterogeneous: 95\% of sessions carry less than 1\,Mbit, while only 1\% of sessions—those exceeding 20\,Mbit—account for more than 70\% of the total traffic volume \cite{ericssonmobilityreport2023}. While most sessions therefore demand modest throughput, traffic exhibits a wide range of payload sizes and latency requirements. Maximizing EE consequently requires adapting both the transmission configuration and sleep-mode operation to the traffic demand.

The analysis in this paper yields several novel results. We discover a new equality at the EE-optimal solution, in which the radiated power equals the total transceiver power. We further show that the EE-optimal signal-to-noise ratio (SNR) collapses to a universal numerical constant of approximately $5.93$\,dB, independent of channel and hardware parameters, and derive closed-form solutions for the optimal ratios of power per bandwidth and power per antenna. We integrate three advanced sleep modes (absolute, deep, and idle) into the optimization framework, characterize the optimal hardware configuration for each mode, and reveal how transition delays and latency constraints jointly shape sleep-mode selection. Finally, we develop algorithms that jointly optimize the design variables under both unconstrained and latency-constrained operation.

\subsection{Contributions}

In this paper, we aim to answer the research questions:
\begin{itemize}
    \item How should the transmit power, bandwidth, and number of antennas be jointly configured to maximize EE? 
    \item Are there any tangible relationships between these parameters at the optimal solution?
    {\color{black}
     \item How do different sleep modes impact the achieved EE when there are strict rate requirements?
     \item How do latency constraints and wake-up transition times determine the sleep-mode selection for data packets?
    }
\end{itemize}

One significant departure from previous papers is our emphasis on analytical scaling behaviors. While prior research primarily focused on optimizing individual parameters, our work extends these models to explore the global optimum. Compared to the conference version \cite{enqvist2024fundamentals}, this paper establishes the existence of an EE-optimal SNR and provides a joint optimization of hardware configuration and sleep mode selection under quality-of-service (QoS) requirements.

The remainder of this paper is organized as follows. Section \ref{sec:system_model} introduces the downlink system model. Section \ref{sec:EERatio} establishes the fundamental scaling behaviors of the EE, analytically deriving optimal ratios for power over bandwidth and power per transmit antenna. In Section \ref{sec:VarOpt}, we optimize the transmit power, bandwidth, and number of antennas, and we propose an algorithm that achieves their joint maximization of the EE. Section \ref{sec:lowrate} extends the framework to incorporate QoS constraints alongside sleep modes, revealing how the BS should adapt its transmission configuration and utilize rush-to-sleep strategies. Section \ref{sec:Latency} further extends the framework with wake-up transition delays and proposes an algorithm that selects the optimal sleep mode and hardware configuration under latency constraints. Finally, Section \ref{sec:conclusion} concludes the paper and suggests future work directions. %

\section{System Model} \label{sec:system_model}

We analyze and optimize the EE of the downlink between a BS using $M$ antennas and a single-antenna UE. The carrier bandwidth is $B$, and the channel is represented by $\mathbf{h}\in \mathbb{C}^M$, where the squared magnitude of each entry is $\beta$. Hence, $\|\mathbf{h}\|^2=M\beta$. The mathematical tractability of the model enables us to evaluate complex scenarios, including varying data rate and latency demands, and different depths of hardware sleep modes in the later sections. %
This is a typical model for a line-of-sight channel. The received downlink signal $y$ is given by
\begin{equation}
    y=\mathbf{h}^{T}\mathbf{p}x+n,
\end{equation}
where $x$ is the data signal, $n\sim \mathcal{N}_\mathbb{C}(0,B N_0 )$ is the independent receiver noise, and $\mathbf{p}\in\mathbb{C}^{M}$ is the unit-norm precoding vector.  Assuming the BS has perfect channel state information (CSI), the capacity of this multiple-input single-output (MISO) channel is achieved by $x \sim \mathcal{N}_\mathbb{C}(0,P)$, where $P$ is the transmit power budget, and the precoding vector $\mathbf{p}=\mathbf{h}^*/\|\mathbf{h}\|$.  The capacity can then be expressed as \cite{Lo1999a,bjornson2017massive}
\begin{equation}
\label{eq:capacity}
    C = B \log_2 \left( 1 + \mathrm{SNR} \right),
\end{equation}
where the signal-to-noise ratio is
\begin{equation} \label{eq:SNR_def}
    \mathrm{SNR} = \frac{M P \beta}{B N_0}.
\end{equation}

Notice that we have made modeling assumptions that enable exact mathematical analysis. However, the analytical insights developed in this paper also hold qualitatively for other channel types, such as Rayleigh fading with variance $\beta$.

\subsection{Power Consumption}

To measure the EE, we must first model the PC. We adopt the model from \cite{lopez2021energy} for single-layer transmission in a single band to a single-antenna UE. The total PC at the BS is
\begin{align}
\label{PC_model}\mathrm{PC}=&P/\kappa+P_\mathrm{FIX}+P_\mathrm{SYN}+D_0 M + D_1 M+\eta C, 
\end{align}
where $\kappa \in (0,1]$ is the PA efficiency and $P_\mathrm{FIX}$ is the load-independent power consumption required for cooling, control signaling, backhaul infrastructure, and baseband processors. $P_\mathrm{SYN}$ is the load-independent power consumed by the local oscillator. $D_0$ is the power consumed by each transceiver chain of the BS (e.g., converters, mixer, filters, etc.). The coefficient $D_1$ determines the per-antenna power consumed by the signal processing and bandwidth-dependent transceiver components (e.g., analog-to-digital converters (ADCs), digital-to-analog converters (DACs), and channel-estimation/precoding hardware). $\eta$ regulates the power consumed by the signal coding at the BS and %
the backhaul signaling, both of which are proportional to the capacity $C$. To simplify the notation and expose the optimization variables, we rewrite \eqref{PC_model} as %
\begin{equation}
\label{eq:PC}
\mathrm{PC}=P/\kappa + \mu+( D_0+\nu B)  M + \eta B \log_2 \left( 1 + \frac{M P \beta}{B N_0} \right),
\end{equation}
where $\mu=P_\mathrm{FIX}+P_\mathrm{SYN}$ denotes the fixed circuit PC from circuitry and synchronization, and $\nu=D_1/B$ is introduced to highlight that the signal processing is carried out on the sampling rate (which is proportional to the bandwidth).

\subsection{Energy Efficiency}

In this paper, we focus on optimizing EE \cite{Verdu1990a,bjornson2017massive} under different constraints. The EE is defined as the amount of data transferred per unit energy (measured in bit/Joule and equivalently bit/s/Watt). Dividing the channel capacity in \eqref{eq:capacity} by the PC in \eqref{eq:PC}, we can define the $\mathrm{EE}$ as
\begin{equation} \label{eq:EE}
\mathrm{EE} =  \frac{B \log_2 \left( 1 + \frac{M P \beta}{B N_0} \right) }{P/\kappa + \mu  + ( D_0+\nu B )M + \eta B \log_2 \left( 1 + \frac{M P \beta}{B N_0}  \right) }. 
\end{equation}
In the following sections \ref{sec:EERatio} and \ref{sec:VarOpt}, we will study the scaling behaviors of the EE with the bandwidth $B$, power $P$, and number of antennas $M$. In particular, we will derive the optimal pairwise ratios for these three design parameters and then develop an algorithm to find the global optimum.

\section{EE-Optimal Parameter Ratios} \label{sec:EERatio}

In this section, we will prove that the optimization variables $B$, $P$, and $M$ satisfy specific ratios at the EE-optimal system operation. These results serve as design guidelines and form the analytical foundation for the joint optimization Algorithm~1 developed in Section \ref{subsection:algorithm}, which is utilized in the remainder of the paper.

\subsection{Power per Antenna}
It is analytically possible to optimize the transmit power per antenna, $P/M$, which yields fundamental insights into BS design. In practice, there may be upper bounds on both parameters that prevent us from achieving the optimal ratio. However, in case the value ranges of $P$ and $M$ are not constrained, or the maximum of the EE in \eqref{eq:EE} is reached without invalidating the bounds $M_\mathrm{max}$ and $P_\mathrm{max}$, the following result is true:

\begin{theorem} \label{PoverM} 
If the solution $(P_\mathrm{opt},M_\mathrm{opt})$ that maximizes the EE in \eqref{eq:EE} or minimizes the PC in \eqref{eq:PC} subject to a capacity constraint for a given value of $B$ satisfies $P_\mathrm{opt} \leq P_\mathrm{max}$ and $M_\mathrm{opt} \leq M_\mathrm{max}$, then the following relation holds:
\begin{equation} \label{eq:optimal-ratio-PM}
\frac{P_\mathrm{opt}}{M_\mathrm{opt}}=\kappa(D_0+\nu B).
\end{equation}
\end{theorem}

\begin{IEEEproof}
$P$ and $M$ jointly determine $C$ strictly through their product $M P$ inside the capacity equation \eqref{eq:capacity}. Therefore, to optimize the PC for any target data rate, the terms in the PC model \eqref{PC_model} that depend on these variables, specifically $\frac{P}{\kappa} + (D_0 + \nu B)M$, must be minimized. We apply the arithmetic mean-geometric mean inequality to these PC terms, which gives
\begin{equation} \label{eq:AMGM}
    \frac{P}{\kappa} + (D_0 + \nu B)M \ge 2\sqrt{\frac{P}{\kappa} (D_0 + \nu B)M}.
\end{equation}
For any required product $M P$, the right side of this inequality represents a lower bound. The minimum PC is achieved if and only if the two terms on the left-hand side of the inequality are equal. Rearranging this equality directly yields the optimal ratio $\frac{P}{M} = \kappa(D_0 + \nu B)$.
\end{IEEEproof}

{\color{black}Theorem \ref{PoverM} has several interesting implications. Because the relationship minimizes the PC for any target product $M P$, it holds universally regardless of the specific higher-level optimization objective.} The right-hand-side in \eqref{eq:optimal-ratio-PM} grows as either $\kappa$, $D_0$, $\nu$, or $B$ increase. It is evident that, as the computational cost $\nu B$ associated with increased bandwidth increases or the PA is of higher quality (i.e., larger $\kappa$), we can afford to transmit more power per antenna when reaching the EE-optimal solution. Moreover, if we can afford to use more antennas to gain higher EE, we should also increase the transmit power to maintain the power per antenna.

Further insights are obtained by rearranging \eqref{eq:optimal-ratio-PM} so that
\begin{equation} \label{eq:optimal-ratio-PM2}
    \frac{P_\mathrm{opt}}{\kappa}=(D_0+\nu B)M_\mathrm{opt}.
\end{equation}
We recognize that $P/\kappa + (D_0+\nu B)M$ appears directly in the PC consumption model in \eqref{eq:PC}. Hence, \eqref{eq:optimal-ratio-PM2} tells us that at the EE-optimal point, the input transmit power $P/\kappa$ is always identical to the power $(D_0+\nu B)M$, i.e., the passive PC in the transceiver chains for all the antennas $D_0 M$ plus the power dissipated in the ADCs and DACs in these transceiver chains $\nu B M$. 

As a side note, the solution in \eqref{eq:optimal-ratio-PM} and \eqref{eq:optimal-ratio-PM2} is strictly true only if $M$ is allowed to attain a non-integer value. From that perspective, it will only be approximately true in practice.

\subsection{Power per Bandwidth Unit}

By dividing the numerator and denominator of \eqref{eq:EE} by $B$, the EE can be rewritten as
\begin{equation} \label{eq:EEdivB}
\mathrm{EE} \!=\!  \frac{\log_2 \left( 1 + \frac{M P \beta}{B N_0} \right) }{P/(\kappa B) + \mu/B  + D_0 M/B+\nu M + \eta \log_2 \left( 1 \!+\! \frac{M P \beta}{B N_0}  \right) }.
\end{equation}
It is apparent from \eqref{eq:EEdivB} that power and bandwidth mostly appear as a ratio $P/B$, which is the power spectral density. The terms $D_0 M /B$ and $\mu/B$ are the only ones that do not fit this structure. However, in the asymptotic limit, these terms are negligible compared to those that depend on both bandwidth and power. This leads to the following key result:

\begin{theorem} \label{maximum-EE-lemma}
When the term $\mu/B  + D_0 M/B$ is negligible, the $\mathrm{EE}$ in \eqref{eq:EEdivB} is maximized when $P$ and $B$ satisfy the ratio 
\begin{equation} \label{eq:optimal-ratio}
\frac{P}{B} = N_0 \frac{e^{u}-1}{M \beta},
\end{equation}
where
\begin{equation}
u= W \left(\frac{\kappa M^2 \beta  \nu}{N_0 e} -\frac{1}{e}\right)+1
\end{equation}
and $W (\cdot)$ denotes the Lambert W function, defined by the equation $x = W(x)e^{W(x)}$ for any $x \in \mathbb{C}$.
\end{theorem}
\begin{IEEEproof}
By defining $z= P/B$, and letting $\mu,D_0 \to 0$, we have that \eqref{eq:EEdivB} can be expressed as
\begin{equation} \label{eq:EE-first-step-SISO-circuit-power2-z}
  \frac{ \log_2 \left( 1 + \frac{M \beta}{N_0} z \right) }{\frac{z}{\kappa} + \nu M + \eta \log_2 \left( 1 + \frac{M \beta}{ N_0} z \right) }.
\end{equation}
The maximum in \eqref{eq:optimal-ratio} is obtained by utilizing \cite[Lem.~3]{bjornson2015a}.
\end{IEEEproof}

A consequence of Theorem \ref{maximum-EE-lemma} is that the maximum EE is achieved by any $P$ and $B$ satisfying \eqref{eq:optimal-ratio} that are sufficiently large to make $\mu/B + D_0 M / B$ negligible. The data rate at the optimum is $C = Bu \log_2(e)$, so $B$ can be chosen freely to deliver any desired rate. If $P$ and $B$ are not limited by external factors, EE and rate are permitted to grow together, with no tradeoff between them as conventionally claimed \cite{Verdu1990a,bjornson2017massive}.

Substituting \eqref{eq:optimal-ratio} into \eqref{eq:EE-first-step-SISO-circuit-power2-z} gives the maximum EE as a function of $M$:
\begin{equation} \label{eq:EE_max_initial}
    \mathrm{EE}_\mathrm{max}(M) = \frac{u \log_2(e)}{N_0\frac{e^u-1}{\kappa M \beta}+\nu M + \eta u \log_2(e)},
\end{equation}
where the effects of $\mu$ and $D_0$ become negligible. From Theorem~\ref{maximum-EE-lemma} and the defining property of the Lambert $W$ function, $u$ satisfies the implicit relation $(u-1)e^u = \frac{\kappa M^2 \beta \nu}{N_0} - 1$. Factoring $\frac{N_0}{\kappa \beta M}$ out of the first two terms in the denominator of \eqref{eq:EE_max_initial} and substituting the implicit relation yields the exact cancellation
\begin{equation}
    \frac{N_0(e^u-1)}{\kappa \beta M} + \nu M  = \frac{N_0}{\kappa \beta M}\, u\, e^u,
\end{equation}
so that
\begin{equation} \label{eq:EE_max}
    \mathrm{EE}_\mathrm{max}(M) = \frac{\log_2(e)}{\frac{N_0}{\kappa \beta} \frac{e^u}{M} + \eta \log_2(e)}.
\end{equation}

We illustrate the behavior of \eqref{eq:EE_max} in Fig.~\ref{fig:MaxEEofM} with the simulation parameters given in Table \ref{tab:table1}. The maximum EE (encircled) is obtained at $M=2$ for $\beta=-100\,$dB, $M=6$ for $\beta=-110\,$dB, and $M=20$ for $\beta=-120\,$dB. The optimal $M$ increases as $\beta$ decreases, scaling by a factor of $10$ when $\beta$ decreases by $20\,$dB. The corresponding optimal power spectral densities are $P/B=79\,$mW/GHz for $\beta=-100\,$dB, $P/B=247\,$mW/GHz for $\beta=-110\,$dB, and $P/B=792\,$mW/GHz for $\beta=-120\,$dB, showing that $P/B$ must increase to overcome larger pathloss. Strikingly, the corresponding SNRs at these three operating points are nearly identical: $6.00\,$dB, $5.71\,$dB, and $6.00\,$dB, respectively. The fact that all three cluster near $6\,$dB despite a $20\,$dB spread in channel gain is not a coincidence; the following theorem characterizes the global optimum of \eqref{eq:EE_max} in closed form and explains the clustering.

\begin{theorem} \label{thm:universal_snr}
The upper bound $\mathrm{EE}_\mathrm{max}(M)$ in \eqref{eq:EE_max} achieves its global maximum at the universal operating SNR
\begin{equation} \label{eq:SNR_star_constant}
    \mathrm{SNR}^\star = e^{u^\star} - 1 \approx 5.93\text{\,dB},
\end{equation}
where $u^\star$ is the numerical constant
\begin{equation} \label{eq:ustar_thmIII}
    u^\star = 2 + W(-2 e^{-2}) \approx 1.5936.
\end{equation}
The corresponding spectral efficiency is the universal constant
\begin{equation} \label{eq:SE_star_constant}
    \mathrm{SE}^\star = u^\star \log_2(e) \approx 2.30 \text{ bit/s/Hz}.
\end{equation}
The optimum is achieved at the antenna count
\begin{equation} \label{eq:M_star_analytical}
    M^\star = \sqrt{\frac{(e^{u^\star} - 1) N_0}{\kappa \beta \nu}},
\end{equation}
and the maximum EE is
\begin{equation} \label{eq:EE_star_ultimate}
    \mathrm{EE}^\star = \frac{\mathrm{SE}^\star}{2 \nu M^\star + \eta\, \mathrm{SE}^\star}.
\end{equation}
\end{theorem}
\begin{IEEEproof}
With $\mu$ and $D_0$ being negligible, Theorem~\ref{PoverM} reduces to $P/M = \kappa \nu B$, so the operating SNR satisfies
\begin{equation}
\mathrm{SNR} = \frac{MP\beta}{BN_0} = \frac{M^2 \kappa \nu \beta}{N_0} = A M^2,
\end{equation}
with $A \triangleq \kappa \beta \nu / N_0$. To find the optimal $M$, we maximize $\mathrm{EE}_\mathrm{max}(M)$ in \eqref{eq:EE_max}. Since the variable part of the denominator is $g(M)\triangleq e^u/M$, the stationary condition $g'(M)=0$ yields $u' = 1/M$. Differentiating the implicit relation $(u-1)e^u = AM^2 - 1$ from Theorem~\ref{maximum-EE-lemma} gives $ue^u u' = 2AM$, and substituting $u'=1/M$ yields
\begin{equation}\label{eq:stationarity}
ue^u = 2AM^2.
\end{equation}
Combining with $(u-1)e^u + 1 = AM^2$ gives $(u-2)e^{u-2} = -2e^{-2}$, which the principal Lambert $W$ branch solves as $u^\star = 2+W(-2e^{-2}) \approx 1.5936$.

From \eqref{eq:stationarity}, $A(M^\star)^2 = e^{u^\star} - 1$, so $\mathrm{SNR}^\star = e^{u^\star} - 1$ and $M^\star$ is given by \eqref{eq:M_star_analytical}. Equation \eqref{eq:SE_star_constant} follows from the capacity expression, and substituting $u^\star$ and $M^\star$ into \eqref{eq:EE_max} yields \eqref{eq:EE_star_ultimate}.
\end{IEEEproof}

if $\eta$ is very small, the maximum EE approximately scales as $\sqrt{\kappa\beta/(\nu N_0)}$, i.e., with the square root of the channel-to-hardware ratio. Theorem \ref{thm:universal_snr} provides a strong design guideline: the EE-maximizing operating point is fixed at $\mathrm{SNR}^* \approx 5.93$\,dB regardless of channel gain, noise, bandwidth, or hardware constants. The corresponding spectral efficiency is $u^\star\log_2(e) \approx 2.30$ bit/s/Hz, which can be achieved using 16-QAM and a forward error correction (FEC) with coding rate $0.58$ \cite{3gpp_ts_38_214}.

\begin{figure}[t!]
    \centering
    \includegraphics[width=0.97\columnwidth]{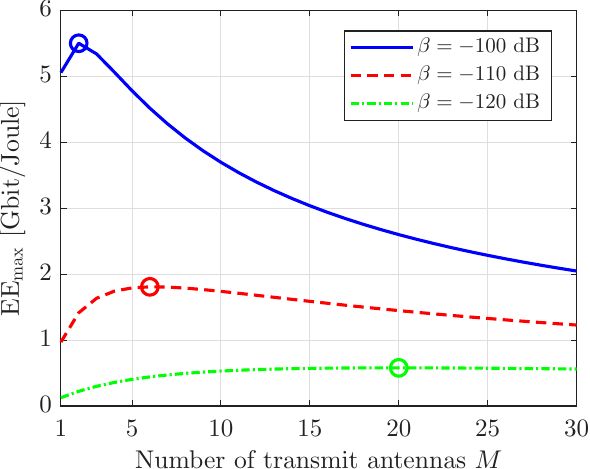}
    \caption{Maximum EE vs. number of transmit antennas $M$. The maximum energy efficiency $\mathrm{EE}_\mathrm{max}$ as defined in \eqref{eq:EE_max} is attained for the finite number $M^\star$ at an optimal ratio $P/B$.}
    \label{fig:MaxEEofM}
    \vspace{-4mm}
\end{figure}

\begin{table}[t!]
  \begin{center}
    \caption{Simulation Parameters}
    \label{tab:table1}
    \begin{tabular}{|l|r|} %
     \hline
      \textbf{Parameter} & \textbf{Value} \\
      \hline
      Passive circuit power: $\mu$ & $1000$\,mW \\
      Transceiver chain power consumption: $D_0$ & $200$\,mW \\
      Sample processing power consumption: $\nu$ & $10^{-10}$\,J/sample \\
      Power amplifier efficiency: $\kappa$ & 0.4 \\
      Computational efficiency: $\eta$ & $10^{-11}$\,J/bit\\
      Maximum bandwidth: $B_\mathrm{max}$ & $10$\,GHz \\
      Maximum power: $P_\mathrm{max}$ & 40\,dBm \\
      Maximum number of transmit antennas: $M_\mathrm{max}$ & $512$ \\
      Receiver noise power spectral density: $N_0$ & $-174$\,dBm/Hz \\
      Channel gain: $\beta$ & $-110$\,dB\\
       \hline
    \end{tabular}
  \end{center}
\end{table}

\section{Variable Optimization} \label{sec:VarOpt}

While the previous section established relations among the optimization variables, this section shows how to optimize the EE with respect to each of the variables $P$, $M$, and $B$ when the others are fixed. These results provide insights into the solution structure and serve as the necessary building blocks for developing a joint optimization algorithm in Section~\ref{subsection:algorithm}. The first result considers optimizing $P$.

\begin{lemma} \label{maximum-EE-P}
The EE in \eqref{eq:EE} for a given $B,M$ is maximized with respect to $P$ by
\begin{equation} \label{eq:optimal-ratio-2}
P = B N_0 \frac{e^{v}-1}{M \beta},
\end{equation}
where
\begin{equation}
v= W \left(\frac{\kappa M \beta(\mu+(D_0+\nu B)M)}{B N_0 e} -\frac{1}{e}\right)+1.
\end{equation}

\end{lemma}
\begin{IEEEproof}
The EE with respect to $P$ has a form that can be directly maximized by using \cite[Lem.~3]{bjornson2015a}.
\end{IEEEproof}

By rearranging \eqref{eq:optimal-ratio-2}, we can once again obtain an expression for the optimal ratio $P/B$. However, in this case, $v$ also depends on $B$, so the result is different from Theorem~\ref{maximum-EE-lemma}.

Next, we optimize $B$ while keeping other parameters fixed. This process can be interpreted as a carrier bandwidth optimization.

\begin{lemma} \label{maximum-EE-B}
The EE in \eqref{eq:EE} is a unimodal function of $B$ (for any fixed $P,M$) that is maximized at a unique $B$ obtained by numerically solving the equation 
\begin{align}\begin{split} \label{eq:numericalB} 
    &\left(\frac{B N_0}{M P \beta}\left( \kappa \mu +D_0 M  \kappa + P   \right)+\kappa \mu + D_0 M \kappa + P \right) \\
    &\times\ln\left(1+\frac{M P \beta}{B N_0}\right)= 
     M \kappa \nu B + \kappa \mu + D_0 M \kappa + P.\end{split}
\end{align}

\end{lemma}
\begin{IEEEproof}
Since $\mathrm{EE}(B)>0$ for $B>0$ and $\lim_{B\to \infty}\mathrm{EE}(B)=\lim_{B\to 0^+}\mathrm{EE}(B)=0$ there exists a positive solution $B_\mathrm{opt}$ that maximizes $\mathrm{EE}(B)$. Furthermore, finding $B_\mathrm{opt}$ by setting $\frac{\partial \mathrm{EE}}{\partial B}=0$ leads to \eqref{eq:numericalB}. This equation has only one solution since its left-hand side goes to $\infty$ for small $B$ but is always decreasing as $B$ grows, and the right-hand side is a positive affine function of $B$. To show that the left-hand side is a monotonically decreasing function, let us take the derivative of \eqref{eq:numericalB} with respect to $B$ and obtain

\begin{align}
  \dfrac{\left({\kappa}{\mu}+D_0M{\kappa}+P\right)\left(BN_0\ln\left(1+\frac{MP{\beta}}{BN_0}\right)-MP{\beta}\right)}{MP{\beta}B}.  
\end{align}
The above function is always non-positive since $\ln(1+x)\leq x$ holds for $x>0$, i.e.,
\begin{align}
    B N_0\ln\left(1+\frac{MP{\beta}}{BN_0}\right)-MP{\beta}\leq 0.
\end{align}
This proves that $\mathrm{EE}(B)$ is a unimodal function of $B$ and that the solution $B_\mathrm{opt}$ can be obtained numerically (e.g., through a bisection search). No closed-form solution exists.
\end{IEEEproof}

The results of Lemma \ref{maximum-EE-P} and \ref{maximum-EE-B} are illustrated in Fig.~\ref{fig:MISOEE}. In this figure, we consider a fixed number of transmit antennas, $ M=20$, and plot the EE as a function of $P$ and $B$. The optimal $P$ for given values of $B,M$ is solved by Lemma \ref{maximum-EE-P} and shown by the black line. The optimal $B$ for given values of $P,M$ is solved by Lemma \ref{maximum-EE-B} and shown by the red line. Furthermore, for large values $P$ and $B$, the two lines converge to the optimal ratio as explained in Theorem \ref{maximum-EE-lemma}.

\begin{figure}[t!]
	\centering \vspace{-2mm}
	\includegraphics[width=\columnwidth]{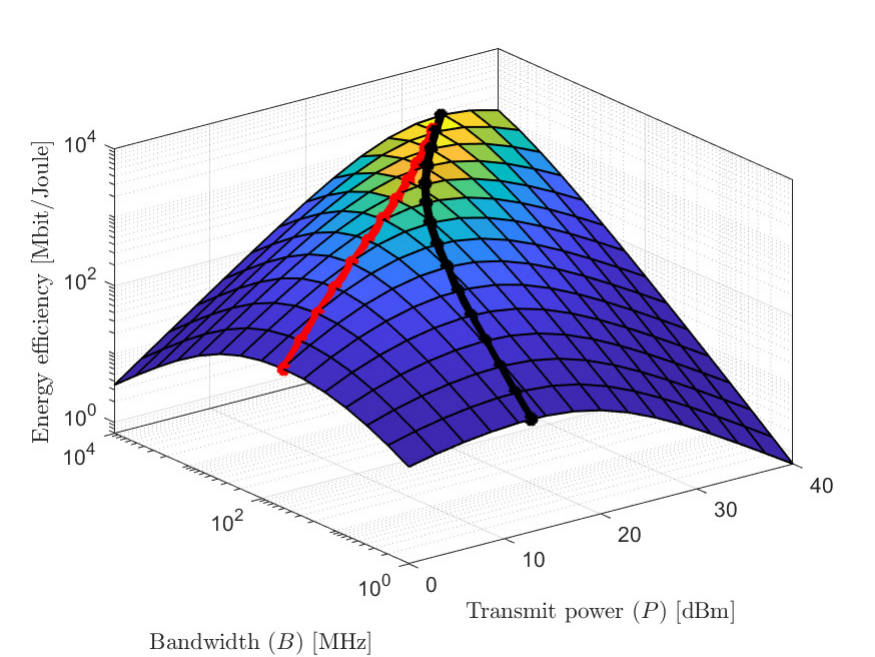}
	\caption{The EE is shown as a function of $P$ and $B$ when $M=20$. In black, we show the optimal $P$ for a given $B$ as in Lemma \ref{maximum-EE-P}. In red, we show the optimal $B$ for a given $P$ as in Lemma \ref{maximum-EE-B}. $P$ and $B$ converge to the optimal ratio $P/B$ given in Theorem 1 when they both take large values.}
	\label{fig:MISOEE} \vspace{-2mm}
\end{figure}

Finally, we optimize $M$ while other parameters are fixed.

\begin{lemma} \label{maximum-EE-M}
The EE in \eqref{eq:EE} for a given $B,P$ is maximized with respect to $M$ by
\begin{equation} \label{eq:optimal-M}
M = B N_0 \frac{e^{w}-1}{P \beta},
\end{equation}
where
\begin{equation}
w= W \left(\frac{P \beta(P/\kappa+\mu)}{B N_0 e(D_0 +\nu B)} -\frac{1}{e}\right)+1.
\end{equation}

\end{lemma}
\begin{IEEEproof}
The EE can be directly maximized by using \cite[Lem.~3]{bjornson2015a}. 
\end{IEEEproof}

In Fig.~\ref{fig:MisoMstar}, we plot the optimal $M$ given by Lemma~\ref{maximum-EE-M} for varying values of $B$ and $P$. We observe that the optimal $M$ spans a wide range, depending on $P$ and $B$.

\begin{figure}[t!]
	\centering \vspace{-2mm}
	\includegraphics[width=\columnwidth]{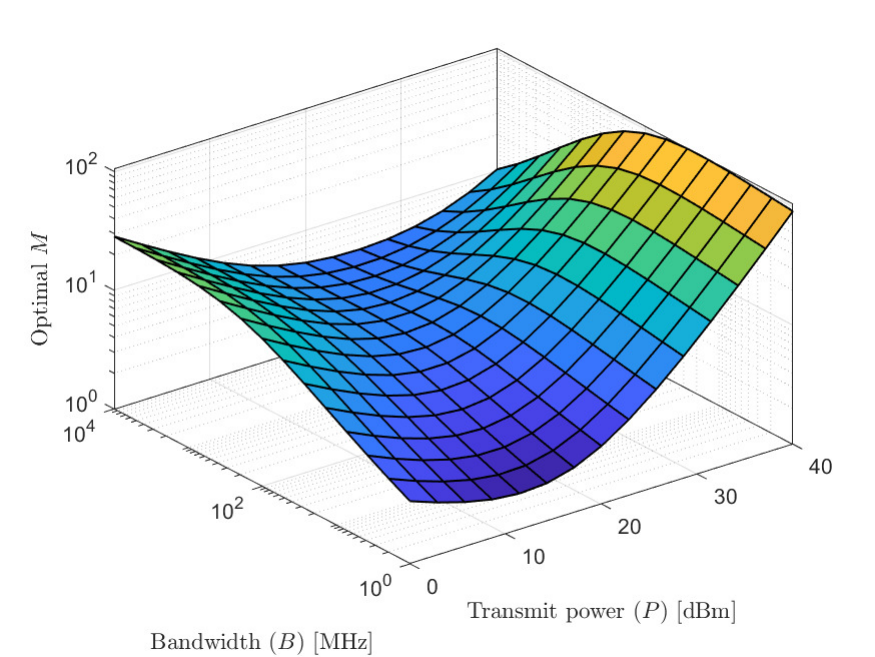}
	\caption{The optimal $M$ as given in Lemma~\ref{maximum-EE-M} as a function of the bandwidth and transmit power. The optimal value is shown on the vertical axis for different $P$ and $B$ values.}
	\label{fig:MisoMstar} \vspace{-2mm}
\end{figure}

\subsection{Computational Efficiency does not Affect the Solution}

As a corollary to the main results, it is interesting to note that the parameter $\eta$ (i.e., the PC constant proportional to the achieved rate) has no impact on the optimal parameters $(P_\mathrm{opt}, B_\mathrm{opt}, M_\mathrm{opt}) $.

\begin{corollary} \label{cor:ignore_eta}
    The optimal solution $(P_\mathrm{opt}, B_\mathrm{opt}, M_\mathrm{opt}) $ that maximizes \eqref{eq:EE} is independent of $\eta$.
\end{corollary}

\begin{IEEEproof}
    We define
    \begin{equation} \label{eq:f}
        f=  \frac{B \log_2 \left( 1 + \frac{M P \beta}{B N_0} \right) }{P/\kappa + \mu  + ( D_0+\nu B )M } ,
    \end{equation}
    which is the EE with $\eta=0$.
    The $\mathrm{EE}$ in \eqref{eq:EE} can then be rewritten as
    \begin{equation}
        \mathrm{EE}=\frac{f}{1+\eta f}.
    \end{equation}
    Equating $\mathrm{EE}'=0$ (with respect to any variable) yields
    \begin{equation}
     \frac{f'(1+\eta f)-f(\eta f')}{(1+ \eta f)^2}=0,
    \end{equation}
    which has the only solution $f'=0$. This implies that $\mathrm{EE}$ is maximized precisely when $f$ is maximized, so the optimal parameters are the same.
\end{IEEEproof}

 The consequence is that optimizing $\mathrm{EE}$ in \eqref{eq:EE} can be facilitated by letting $\eta=0$. A similar observation was made in \cite{bjornson2015a} but for a different system model.

\subsection{Algorithm for Optimizing the EE}
\label{subsection:algorithm}

An algorithm that utilizes our previous results in Lemmas 1-3 and that converges to the optimal solution is provided in this section. In practical deployments, the BS is restricted by a maximum transmit power $P_\mathrm{max}$ to comply with safety regulations, reduce out-of-band emissions, and hardware limitations \cite{3gpp.38.104.r19}. The available spectrum is bounded by a maximum bandwidth $B_{\max}$ allocated by regulatory authorities. The number of antennas is physically limited to $M_{\max}$.

The goal is to solve the following joint EE maximization problem:

\begin{equation} \label{eq:optimization}
\begin{aligned}
 \underset{P,B,M}{\textrm{maximize}} \quad & \mathrm{EE}\\
\textrm{subject to} \quad & 0 < P \leq P_\mathrm{max},\\
& 0 < B \leq B_\mathrm{max}, \\
 \quad & M\in \{1,\ldots ,M_\mathrm{max}\},
\end{aligned}
\end{equation}
with the $\mathrm{EE}$ defined as in \eqref{eq:EE}.

The following result on whether $P$ or $B$ should be maximized is needed.%
\begin{theorem}\label{eitheror}
The constrained EE maximization problem in \eqref{eq:optimization} is solved at the boundary where $P=P_{\max}$ or $B=B_{\max}$. Specifically, the optimal solution lies on the $B=B_{\max}$ boundary if
\begin{equation}
    \frac{P_{\max}}{\kappa} \geq (D_0 + \nu B_{\max})\widehat{M},
\end{equation}
where $\widehat{M}$ is the optimal number of antennas at $(P_{\max}, B_{\max})$ given by Lemma \ref{maximum-EE-M}. Otherwise, the solution lies on the $P=P_{\max}$ boundary.
\end{theorem}

\begin{IEEEproof}
For any fixed ratio $P/B$, scaling $P$ and $B$ jointly by a factor $c>1$ leaves the SNR unchanged and strictly increases the EE because the fixed power $\mu + D_0 M$ becomes diluted. Consequently, no interior optimum exists and the global maximum must lie on the boundary where either $P=P_{\max}$ or $B=B_{\max}$. Since the EE is jointly strictly pseudo-concave in $(P,B)$ \cite{ZapponeNowPublishers2015}, evaluating $\frac{\partial \mathrm{EE}}{\partial P} \leq 0$ at the corner $(P_{\max},B_{\max})$ indicates that decreasing $P$ increases the EE, placing the optimum on the $B=B_{\max}$ boundary. Otherwise, it lies on the $P=P_{\max}$ boundary. Ignoring $\eta$ via Corollary 1 and expanding $\frac{\partial \mathrm{EE}}{\partial P} \leq 0$ at $(P_{\max},B_{\max},\widehat{M})$ yields
\begin{equation} \label{eq:boundary_ineq}
\frac{\kappa \widehat{\gamma}}{P_{\max}} \big(\mu + \widehat{M}(D_0 + \nu B_{\max})\big) \leq (1+\widehat{\gamma})\ln(1+\widehat{\gamma}) - \widehat{\gamma},
\end{equation}
where we define $\widehat{\gamma} = \frac{\widehat{M} P_{\max} \beta}{B_{\max} N_0}$. Because $\widehat{M}$ optimizes the EE at this coordinate by Lemma \ref{maximum-EE-M}, it satisfies $\frac{\partial \mathrm{EE}}{\partial M} = 0$. Expanding this derivative and equating to zero yields the algebraic identity:
\begin{equation} \label{eq:M_identity}
(1+\widehat{\gamma})\ln(1+\widehat{\gamma}) - \widehat{\gamma} = \frac{\widehat{\gamma}}{\widehat{M}} \frac{P_{\max}/\kappa + \mu}{D_0 + \nu B_{\max}}.
\end{equation}
Substituting \eqref{eq:M_identity} into \eqref{eq:boundary_ineq}, dividing by $\widehat{\gamma}$, and multiplying by $\widehat{M} \frac{P_{\max}}{\kappa} (D_0 + \nu B_{\max})$ produces the inequality:
\begin{equation}
\big(\widehat{M}(D_0 + \nu B_{\max})\big)^2 + \mu \widehat{M}(D_0 + \nu B_{\max}) \leq \left(\frac{P_{\max}}{\kappa}\right)^2 + \mu \frac{P_{\max}}{\kappa}.
\end{equation}
Since $f(x) = x^2 + \mu x$ is monotonically increasing for $x > 0$, this inequality holds if and only if $\widehat{M}(D_0 + \nu B_{\max}) \leq P_{\max}/\kappa$ which is the condition stated in the theorem.
\end{IEEEproof}

The condition in this theorem compares the transceiver power consumption versus the radiated power consumption at $(P_\mathrm{max},B_\mathrm{max})$.

We propose Algorithm 1 to solve \eqref{eq:optimization}. The starting configuration $(P_\mathrm{max}, B_\mathrm{max}, \widehat{M})$ is principled: Theorem \ref{eitheror} guarantees the optimum lies on the boundary emanating from this corner, and the convergence tolerance $\delta$ is set proportional to the initial EE so the stopping criterion is scale-invariant. If the Theorem \ref{eitheror} inequality holds, the optimum lies on the $B=B_\mathrm{max}$ boundary, and we invoke Lemma \ref{maximum-EE-P} to find the optimal $P$. Otherwise, the optimum lies on the $P=P_\mathrm{max}$ boundary, and we invoke Lemma \ref{maximum-EE-B} to compute the optimal $B$. We then use Lemma \ref{maximum-EE-M} to update the number of antennas $M$ based on the obtained $P$ and $B$, capped at $M_\mathrm{max}$, and alternate this update with the boundary optimization until the EE increase between iterations falls below $\delta$. Because the number of antennas must be an integer, the final step considers the two closest integers to the converged $M$, recomputes the corresponding optimal $P$ and $B$, and selects the configuration with the highest EE. The objective is strictly pseudo-concave in each variable individually \cite{ZapponeNowPublishers2015} and bounded from above by \eqref{eq:EE_star_ultimate}, so it possesses a unique global maximum. Since each step of the alternating optimization strictly maximizes one dimension, the EE monotonically increases and convergence is assured. The updates of $P$ and $M$ are shown in Fig.~\ref{fig:algo1}.

\begin{algorithm}[b!]
\label{algo1}
\caption{Joint EE Maximization}
\begin{algorithmic}[1]
\renewcommand{\algorithmicrequire}{\textbf{Input:}}
\renewcommand{\algorithmicensure}{\textbf{Output:}}
\REQUIRE $\beta, N_0, \mu, \nu, D_0, \kappa, P_\mathrm{max}, B_\mathrm{max}, M_\mathrm{max}$
\ENSURE $P, B, M, \mathrm{EE}$
\STATE \textit{use Lemma \ref{maximum-EE-M}} to initialize $M=\widehat{M}$ at $(P_\mathrm{max}, B_\mathrm{max})$
\STATE \textit{set tolerance}: $\delta = 10^{-3}\, \mathrm{EE}(P_\mathrm{max}, B_\mathrm{max}, \widehat{M})$
\STATE initialize $\mathrm{EE}_\mathrm{prev} \gets 0$
\STATE initialize $\Delta \gets \infty$

\STATE \textit{evaluate boundary condition from Theorem \ref{eitheror}}:
\IF {$\widehat{M} < P_\mathrm{max}/[\kappa(D_0 + \nu B_\mathrm{max})]$}
    \STATE \textit{optimum is on the bandwidth boundary}: set $B=B_\mathrm{max}$
\ELSE
    \STATE \textit{optimum is on the power boundary}: set $P=P_\mathrm{max}$
\ENDIF

\WHILE{$\Delta \geq \delta$}
  \IF {$B == B_\mathrm{max}$}
      \STATE \textit{use Lemma \ref{maximum-EE-P}} to optimize $P(B_\mathrm{max}, M)$
  \ELSE
      \STATE \textit{use Lemma \ref{maximum-EE-B}} to optimize $B(P_\mathrm{max}, M)$
  \ENDIF
  \STATE \textit{use Lemma \ref{maximum-EE-M}} to update $M(P, B)$, capped at $M_\mathrm{max}$
  \STATE update $\Delta \gets \mathrm{EE}(P, B, M) - \mathrm{EE}_\mathrm{prev}$ and $\mathrm{EE}_\mathrm{prev} \gets \mathrm{EE}(P, B, M)$
\ENDWHILE

\STATE \textit{compare the closest integers to $M$}
\STATE update $(P,B)$ for $M=\lfloor M \rfloor$ and $M=\lceil M \rceil$ as in lines 12--16
\STATE \textit{identify the optimal solution as the one of those two candidate solutions that attains the highest EE}
\RETURN $P, B, M, \mathrm{EE}$
\end{algorithmic} 
\end{algorithm}

\begin{figure}[t!]
	\centering \vspace{-4mm}
	\includegraphics[width=\columnwidth]{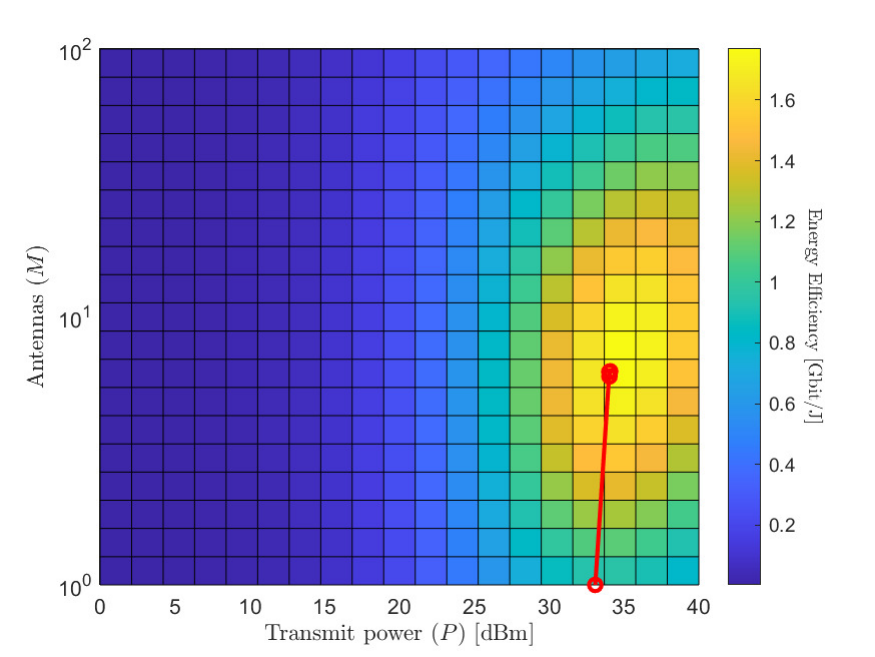}
	\caption{The EE is shown as a function of $P$ and $M$. It converges to the optimum point as detailed in Algorithm 1. Encircled in red we have the algorithm's updates of $P$ and $M$ and the corresponding $\mathrm{EE}$.}
	\label{fig:algo1} \vspace{-2mm}
\end{figure}

\section{Rate Constraints and Sleep Modes} \label{sec:lowrate}

The maximum EE is achieved at a particular data rate, which might not match the rate required by the application that generates that data. Hence, we must extend our framework to determine the most energy-efficient way to deliver the required QoS $R$. If the QoS requirement is low, then the BS might not need to transmit continuously. Instead, it can achieve an average rate of  $R$ by switching between operating at a higher rate and exploiting idle periods through advanced sleep modes. This strategy is known as ``rush-to-sleep" and allows the hardware to dynamically scale its PC based on real-time traffic demands. %
This section integrates three distinct advanced sleep modes into our optimization framework from Section \ref{subsection:algorithm} to derive novel optimal operating points for energy-efficient transmissions when the BS can exploit idle periods.

\subsection{Sleep Mode Power Consumption Model}

Our PC model in \eqref{eq:PC} is valid when the BS is actively transmitting data. When the BS is sleeping (or waiting to transmit), it consumes only the constant sleep power 
\begin{equation}
    P_{s}=\delta_s(\mu+D_0M)
\end{equation}
for some sleep level parameter $\delta_s\in[0,1]$ that depends on the hardware capabilities of the BS and the depth of the sleep \cite{debaillie2015flexible, 3gpp.38.864}. The deeper the sleep mode, the less power is consumed. The value of $P_{s}$ is lower than or equal to the PC of the BS in active mode since $\delta_s \leq 1$, with equality if and only if the BS is active without transmitting or processing. We furthermore assume that the transition between active mode and sleep mode is instantaneous and consumes no power. In this paper, we consider three distinct sleeping modes with varying values of $\delta_s$ as described below. These modes are chosen for their theoretical significance.\footnote{Note that 3GPP also specifies other values that might be used in this context \cite{3gpp.38.864}.} In practice, the possible choices of sleep modes depend heavily on the traffic (through wake-up times) and are restricted by latency and scheduling decisions in order to satisfy the QoS requirements. This is investigated further in Section \ref{sec:Latency}, while we consider all the following modes to be available in this section:

\begin{itemize}
    \item \textbf{Absolute sleep:} $P_{s}=0$ and $\delta_s=0$; i.e., the BS consumes no power when sleeping.
    \item \textbf{Deep sleep:} $P_{s}=\mu$ and $\delta_s= \frac{\mu}{\mu+D_0M}$; i.e., the BS only consumes its passive circuit power when sleeping.
    \item \textbf{Idle mode:} $P_{s}=\mu+D_0M$ and $\delta_s=1$; i.e., the BS consumes the passive circuit power needed to power itself and its transceivers.
\end{itemize}

Let $\alpha \in [0,1]$ denote the activity factor of the BS, defined as the proportion of time the BS is actively transmitting. Consequently, the BS remains in sleep mode for the remaining $(1-\alpha)$ of the time. The average $\mathrm{PC}$ becomes
\begin{align} \nonumber
    \mathrm{PC}=&\alpha\underbrace{\left(\frac{P}{\kappa} +\mu+ (D_0+\nu B)M +\eta  B \log_2\left(1+\frac{MP\beta}{BN_0}\right)\right)}_{P_a}\\
    &+(1-\alpha)\underbrace{\delta_s(\mu+D_0M)}_{P_s}. \label{eq:PaandPs}
\end{align}
The required rate $R$ uniquely determines the activity factor as 
\begin{equation} \label{eq:activityfactor}
    \alpha(R) = \frac{R}{B \log_2\left(1+\frac{MP\beta}{BN_0}\right)}.
\end{equation}

\subsection{Problem Formulation}

We want to minimize the PC while satisfying the rate constraint by solving
\begin{equation} \label{eq:PC_rateopt}
\begin{aligned}
 \underset{P,B,M,\alpha}{\textrm{minimize}} \quad &  \mathrm{PC}\\
\textrm{subject to} \quad & 0 < P \leq P_\mathrm{max}, \, 0 < B \leq B_\mathrm{max},
& \\ \quad & M \in \{1,\dots, M_\mathrm{max}\},\, \alpha \leq 1,
& \\ &\alpha B \log_2\left(1+\frac{MP\beta}{BN_0}\right)=R.
\end{aligned}
\end{equation}

Inserting $\alpha$ from \eqref{eq:activityfactor} into the PC defined in \eqref{eq:PaandPs} and rewriting, the problem \eqref{eq:PC_rateopt} becomes equivalent to maximizing the EE in \eqref{eq:1337} as shown at the top of the next page.

\begin{figure*}
\begin{equation} \label{eq:1337}
\begin{aligned}
 \underset{P,B,M}{\textrm{maximize}} \quad & \frac{B \log_2\left(1+\frac{MP\beta}{BN_0}\right)}{\frac{P}{\kappa}+(1-\delta_s)(\mu+ D_0M) +\nu BM+\left(\eta+\frac{\delta_s (\mu+D_0 M)}{R}\right)B\log_2\left(1+\frac{MP\beta}{B N_0}\right)} \\
\textrm{subject to} \quad & 0 < P \leq P_\mathrm{max}, \, 0 < B \leq B_\mathrm{max}, \,
M \in \{1,\dots, M_\mathrm{max}\}, \,  B\log_2\left(1+\frac{MP\beta}{B N_0}\right) \geq R.
\end{aligned}
\end{equation}
\hrulefill 
\end{figure*}

We now show that utilizing sleep modes increases the EE for all rates satisfying $R<R_\mathrm{opt}$, where
$R_\mathrm{opt}$ denotes the rate resulting from the optimal configuration obtained via Algorithm 1 for the EE maximization in \eqref{eq:optimization}.

\subsection{Absolute Sleep}
Under the assumption of absolute sleep, where $\delta_{s}=0$ (and equivalently $P_s=0$), and by invoking Corollary \ref{cor:ignore_eta} to neglect $\eta$, the EE objective function in \eqref{eq:1337} reduces to
\begin{equation} \label{eq:optimization4andhalf}
\mathrm{EE}_\mathrm{absolute}(P,B,M)=\frac{B \log_2\left(1+\frac{MP\beta}{BN_0}\right)}{\frac{P}{\kappa} + \mu + (D_0+\nu B)M}.%
\end{equation}
This is equivalent to the non-rate-constrained EE problem in \eqref{eq:EE}. Consequently, the optimal operating configuration, achieved EE, and active transmit rate are identical to the solutions obtained by optimizing \eqref{eq:EE}, implying that
\begin{align}&(P_\mathrm{absolute},B_\mathrm{absolute},M_\mathrm{absolute},\mathrm{EE}_\mathrm{absolute}) \nonumber \\ & \quad \quad =(P_\mathrm{opt},B_\mathrm{opt},M_\mathrm{opt},\mathrm{EE_\mathrm{opt}}),
\end{align}
which can be obtained by Algorithm 1. The rate during transmission is $R_\mathrm{absolute}=R_\mathrm{opt}$ and the average rate is $R=\alpha R_\mathrm{opt}$.

\subsection{Deep Sleep}
For the deep sleep regime characterized by $P_s = \mu$ (equivalently $\delta_s = \mu/(\mu + D_0 M)$), the fixed circuit power $\mu$ is paid uniformly across active and sleep periods and therefore drops out of the EE objective. Applying Corollary~\ref{cor:ignore_eta} to neglect $\eta$, the objective in \eqref{eq:1337} reduces to
\begin{equation} \label{eq:optimization9}
\mathrm{EE}_\mathrm{deep}(P, B, M) = \frac{B \log_2\left(1 + \frac{MP\beta}{BN_0}\right)}{P/\kappa + (D_0 + \nu B)M}.
\end{equation}
This objective retains $D_0$ but admits a closed-form solution.

\begin{theorem} \label{thm:deep_sleep}
For deep sleep, when the solution is obtained at $B = B_\mathrm{max}$ and $M \le M_\mathrm{max}$, the EE in \eqref{eq:optimization9} has a closed-form maximum. The optimal parameters are
\begin{align}
B_\mathrm{deep} &= B_\mathrm{max}, \\
M_\mathrm{deep} &= \sqrt{\frac{(e^{u^\star} - 1) N_0}{\kappa \beta (\nu + D_0/B_\mathrm{max})}}, \\
P_\mathrm{deep} &= \kappa(D_0 + \nu B_\mathrm{max}) M_\mathrm{deep},
\end{align}
where $u^\star \approx 1.5936$ is the constant from Theorem~\ref{thm:universal_snr}. The resulting operating SNR and active transmission rate are
\begin{align}
\mathrm{SNR}_\mathrm{deep} &= e^{u^\star} - 1, \\
R_\mathrm{deep} &= B_\mathrm{max}\, u^{\star} \log_2(e).
\end{align}
\end{theorem}

\begin{IEEEproof}
By Theorem~\ref{eitheror}, $B_\mathrm{deep} = B_\mathrm{max}$ ensures the power constraint is satisfied. By Theorem~\ref{PoverM}, $P = \kappa(D_0 + \nu B)M$ at the optimum. Substituting this into \eqref{eq:optimization9} reduces the objective to a function of $M$ alone:
\begin{equation}
\mathrm{EE}_\mathrm{deep}(M) = \frac{B_\mathrm{max} \log_2\left(1 + \frac{\kappa \beta (\nu + D_0/B_\mathrm{max})}{N_0} M^2 \right)}{2(D_0 + \nu B_\mathrm{max}) M}.
\end{equation}
Setting $\mathrm{EE}'_\mathrm{deep}(M) = 0$ yields
\begin{equation}
A M^2_\mathrm{deep} - \tfrac{1}{2}(1 + A M^2_\mathrm{deep})\ln(1 + A M^2_\mathrm{deep}) = 0,
\end{equation}
where $A \triangleq \kappa \beta (\nu + D_0/B_\mathrm{max})/N_0$. With $\xi = 1 + A M^2$ this becomes $(-2/\xi)e^{-2/\xi} = -2e^{-2}$, which the principal Lambert $W$ branch solves as $\xi = -2/W(-2e^{-2}) \approx 4.9216$. Noting that $\xi = e^{u^\star}$ with $u^\star = 2 + W(-2e^{-2}) \approx 1.5936$ from Theorem~\ref{thm:universal_snr}, we substitute back to obtain $M_\mathrm{deep} = \sqrt{(e^{u^\star} - 1)/A}$, and $P_\mathrm{deep} = \kappa(D_0 + \nu B_\mathrm{max}) M_\mathrm{deep}$ then follows from Theorem~\ref{PoverM}. The SNR follows from its definition as $\mathrm{SNR}_\mathrm{deep} = A M^2_\mathrm{deep} = e^{u^\star} - 1$, and the active rate becomes $R_\mathrm{deep} = B_\mathrm{max}\log_2(e^{u^\star}) = B_\mathrm{max}\, u^\star \log_2(e)$.
\end{IEEEproof}

The fact that Theorem~\ref{thm:deep_sleep} yields the same numerical constant $u^\star$ as Theorem~\ref{thm:universal_snr} is not a coincidence. The universal-SNR result of Theorem~\ref{thm:universal_snr} required $\mu$ and $D_0$ to be negligible. Theorem~\ref{thm:deep_sleep} admits a closed-form even with $D_0 > 0$ and $\mu > 0$. The convergence of the two theorems on the same optimal SNR shows that $5.93$ dB is a robust EE operating point.%

This configuration is reached by the rush-to-sleep strategy: transmitting briefly at $R_\mathrm{deep}$ and sleeping the rest of the time. A BS that can fully power down avoids paying $\mu$ during idle periods and therefore reaches higher EE for a given rate $R$. Deeper sleep capability translates directly into higher achievable EE. Under absolute sleep, the universal operating point is reached only asymptotically as $\mu/B \to 0$; for finite $B_\mathrm{max}$, the absolute-sleep optimum sits at a strictly higher SNR, since the additional fixed cost in \eqref{eq:optimization4andhalf} shifts the optimal balance toward larger transmit power.

\subsection{Idle Mode}

For idle mode, $P_s = \mu + D_0 M$ and $\delta_s = 1$. Applying Corollary \ref{cor:ignore_eta} to neglect $\eta$ reduces the EE objective in \eqref{eq:1337} to
\begin{equation} \label{eq:EEmicro}
    \mathrm{EE}_\mathrm{idle} = \frac{B \log_2\left(1+\frac{MP\beta}{BN_0}\right)}{\frac{P}{\kappa} + \nu B M + \frac{D_0 M}{R} B \log_2\left(1+\frac{MP\beta}{BN_0}\right)}.
\end{equation}
Unlike absolute and deep sleep, the rate $R$ now appears explicitly in the denominator through the term $D_0 M / R$, so the optimal hardware configuration depends on the rate requirement. For any given $M$, the optimal $P$ and $B$ are still obtained from Lemmas \ref{maximum-EE-P} and \ref{maximum-EE-B}, but $M$ itself must be found numerically. A simple sweep over $M = 1, \dots, M_\mathrm{max}$ with $P$ and $B$ optimized at each step suffices: the numerator grows logarithmically in $M$ while the denominator grows at least linearly, so the EE is unimodal in $M$ and the optimum is identified at the point just before EE degradation. At low rates, the $D_0 M / R$ term dominates the denominator, which favors using fewer antennas than deep sleep would.

\begin{figure}[t!]
    \centering
    \subfloat[EE of sleep modes and sustained transmission for different rate requirements.]{
        \includegraphics[width=\columnwidth]{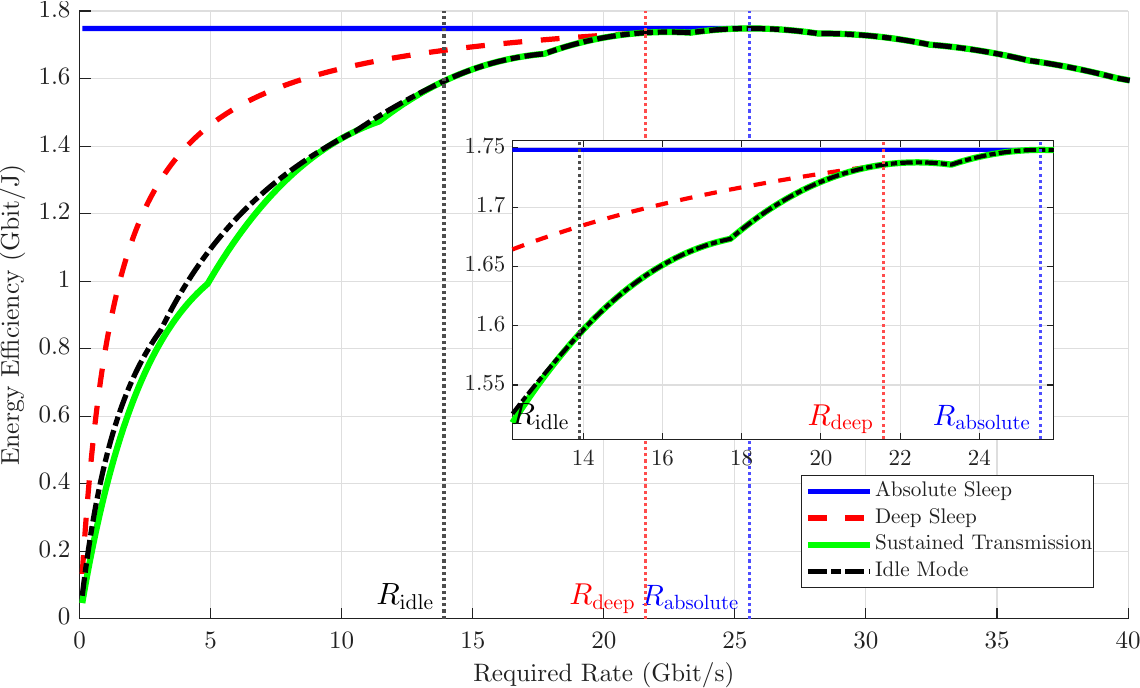}
        \label{fig:sleepEEvsR}
    }
        \\ \vspace{-2mm} %
    \subfloat[Operating SNR of sleep modes and sustained transmission for different rate requirements.]{
        \includegraphics[width=\columnwidth]{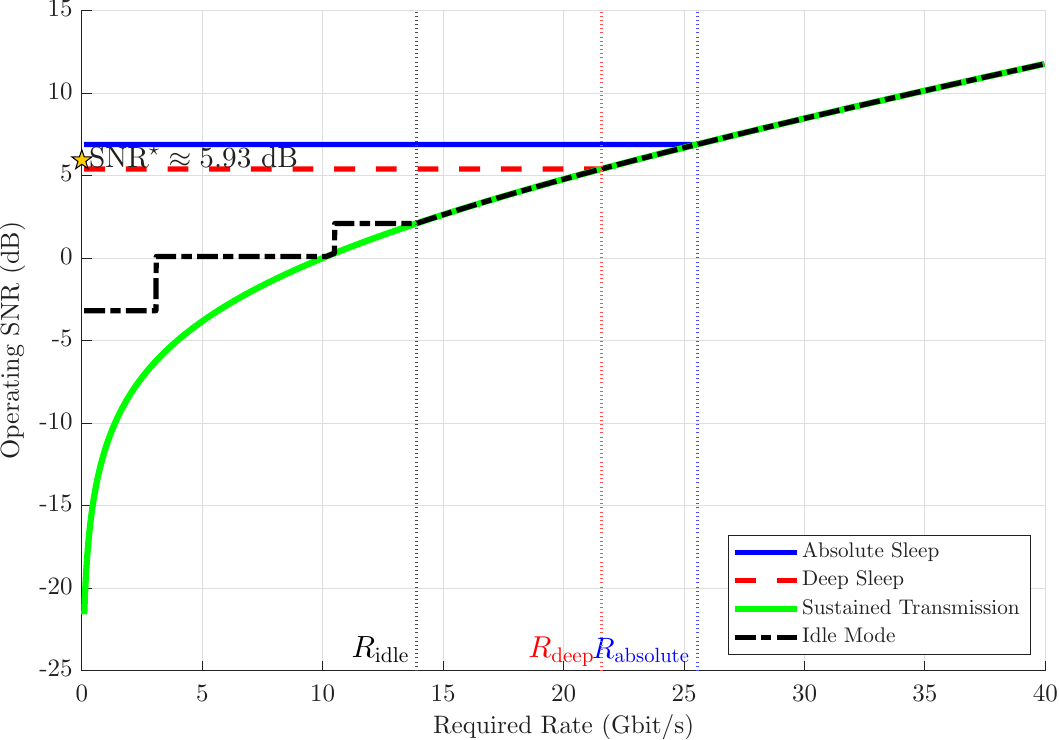}
        \label{fig:sleepSNRvsR}
    }
    \\ \vspace{-2mm} %
    \subfloat[Optimal activity factors $\alpha$ of sleep modes for different rate requirements.]{
        \includegraphics[width=\columnwidth]{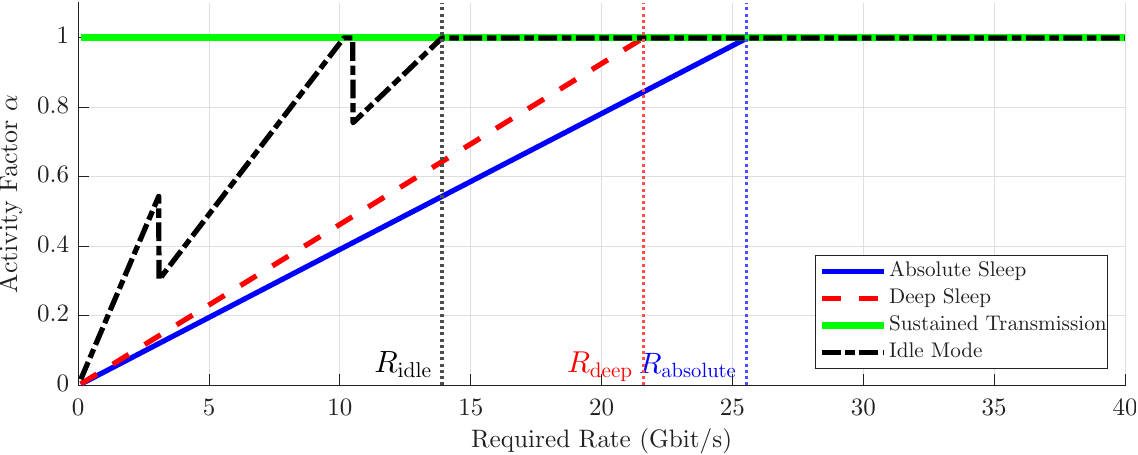}
        \label{fig:sleepALPHAvsR}
    }
    \caption{BS EE (a), operating SNR (b) and activity factor (c) across varying data rate requirements.}
    \label{fig:combined_performance}
    \vspace{-4mm}
\end{figure}

Fig.~\ref{fig:sleepEEvsR} plots the EE for all three sleep modes as a function of the required rate $R$. We also include in green a ``sustained transmission'' baseline: the configuration the BS uses when $R \ge R_\mathrm{opt}$ and sleep is no longer feasible, derived in Section~\ref{sec:highrate} below. For idle mode, sustained transmission is utilized for rates exceeding $R_\mathrm{idle}$, which is evident as their respective curves coincide. Similarly, for deep and absolute sleep, the BS transitions entirely to sustained transmission for rates greater than $R_\mathrm{deep}$ and $R_\mathrm{absolute}$, respectively. These exact transition points and the merging of the curves are clearly visible in the zoomed-in inset.

The corresponding operating SNRs are shown in Fig.~\ref{fig:sleepSNRvsR}. Deep sleep operates at the universal $\mathrm{SNR}^\star \approx 5.93\,$dB from Theorem~\ref{thm:universal_snr} (yellow star, observed here at $5.71\,$dB due to integer quantization), while absolute sleep operates at a strictly higher SNR because of the additional fixed cost $\mu$ in \eqref{eq:optimization4andhalf}. Once the rate exceeds the threshold of a given sleep mode, that mode's curve aligns with the sustained-transmission curve.%

The corresponding activity factor $\alpha$ is shown in Fig.~\ref{fig:sleepALPHAvsR}. The sawtooth pattern observed in $\alpha$ for idle mode arises directly from the rate-dependent integer optimization of the active antennas, $M$. Whenever $R$ crosses a threshold that justifies activating an additional antenna, the active transmission rate jumps upward. Because $\alpha$ is the ratio of the required rate to the active rate, the activity factor drops sharply in response. Between these thresholds, $M$ remains fixed and $\alpha$ climbs steadily until the next activation. Absolute and deep sleep do not exhibit this pattern, as their optimal $M$ remains constant. When the activity factor of any sleep mode reaches one, the BS defaults to sustained transmission and all curves achieve identical EE.

\begin{figure}[t!]
	\centering \vspace{-0mm}
	\includegraphics[width=0.9\columnwidth]{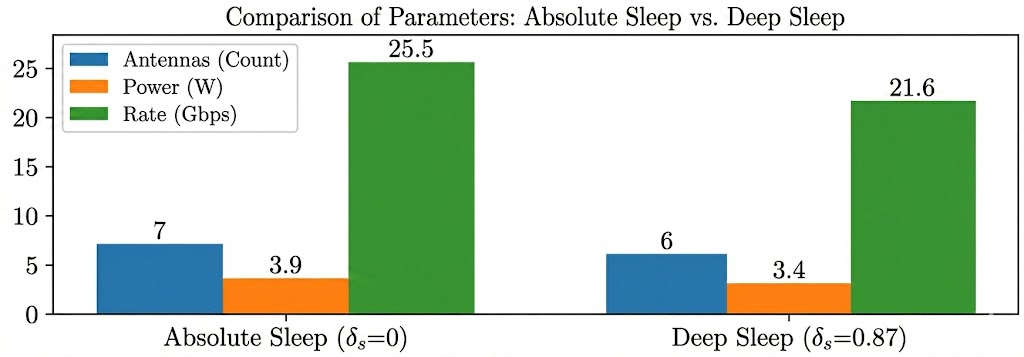}
	\caption{The optimal number of antennas, transmit power and  rate which maximizes the EE in absolute and deep sleep modes.} %
	\label{fig:MPRsleep} \vspace{-2mm}
\end{figure}

The corresponding hardware configurations for absolute and deep sleep are compared in Fig.~\ref{fig:MPRsleep}: absolute sleep transmits at higher $P$, with more antennas, and at a higher active rate, consistent with its higher operating SNR in Fig.~\ref{fig:sleepSNRvsR}. Across all three modes, EE is maximized by the rush-to-sleep strategy of transmitting briefly at the optimal active configuration and sleeping the rest of the time. All three sleep-capable modes outperform a BS that can only use sustained transmission.

\subsection{Sustained Transmission} \label{sec:highrate}
The sleep-mode strategies above all rely on the activity factor $\alpha$ being strictly less than one, which requires the required rate being below the active rate of the chosen sleep mode. As $R$ grows, $\alpha$ approaches one and the benefit of any sleep mode vanishes. For rates $R \ge R_\mathrm{opt}$, sleep is no longer feasible and the BS must transmit continuously at the required rate. The EE problem in \eqref{eq:1337} then reduces to minimizing the active power consumption subject to the rate constraint:
\begin{equation} 
\begin{aligned} 
\underset{P,B,M}{\textrm{minimize}} \quad & \frac{P}{\kappa} + \mu + (D_0+\nu B)M + \eta R\\
\textrm{subject to}
 \quad & B \log_2\left(1+\frac{MP\beta}{BN_0}\right) =R, \\
\quad & 0 < P \leq P_\mathrm{max}, \, 0 < B \leq B_\mathrm{max}, \\
 \quad & M\in \{1,\ldots ,M_\mathrm{max}\}. \label{eq:highrateoriginal}
\end{aligned}
\end{equation}
Provided that the optimal $P$ and $M$ are below their respective upper bounds, Theorem~\ref{PoverM} dictates that $P = \kappa(D_0 + \nu B)M$. Substituting this relationship into the rate constraint of \eqref{eq:highrateoriginal} and fixing the bandwidth at its maximum limit $B = B_\mathrm{max}$ yields the closed-form sustained transmission solution
\begin{align}
M_\mathrm{sust} &= \sqrt{\frac{(2^{R/B_\mathrm{max}} - 1)N_0}{\kappa\beta(\nu + D_0/B_\mathrm{max})}}, \\
P_\mathrm{sust} &= \kappa(D_0 + \nu B_\mathrm{max}) M_\mathrm{sust}.
\end{align}
The corresponding SNR is $\mathrm{SNR}_\mathrm{sust} = 2^{R/B_\mathrm{max}} - 1$, which exceeds $\mathrm{SNR}^\star=e^{u^\star} - 1$ whenever the rate constraint forces $R > R_\mathrm{deep}$. The resulting EE and SNR are plotted in green in Fig.~\ref{fig:sleepEEvsR} and \ref{fig:sleepSNRvsR}. The exponential dependence on $R$ in $P_\mathrm{sust}$ produces a drop in EE as the rate requirement grows beyond the active rate of the EE-optimal sleep configuration, reflecting the cost of pushing the system above its EE-optimal operating point. 

\section{Latency Constraints}\label{sec:Latency}

In practical wireless networks, downlink traffic consists of bursty sessions with a wide range of payload sizes and latency requirements. The strategy for managing the BS PC is therefore dictated by the traffic characteristics and their associated QoS constraints. This section extends our framework to account for state transition delays and shows how the sleep modes of Section \ref{sec:lowrate} can be combined with the optimization of Section \ref{sec:VarOpt} to minimize EC under deadline constraints.

\subsection{Problem Formulation} \label{subsec:latency_problem}

Adopting the PC model and sleep modes from Section \ref{sec:lowrate}, we extend our analysis to account for state transition delays. Let $\tau_k \geq 0$ denote the mode-specific wake-up transition time for sleep mode $k \in \{\text{absolute, deep, idle}\}$. Based on the 3GPP specifications \cite{3gpp.38.864}, we define $\tau_\mathrm{absolute} = 50$\,ms for absolute sleep, $\tau_\mathrm{deep} = 6$\,ms for deep sleep, and an instantaneous transition $\tau_\mathrm{idle} = 0$\,ms for idle mode. Transitioning into a sleep state incurs no energy cost and is assumed to be instantaneous \cite{kundu2025toward}, but it could also be considered part of the transition time.

Two distinct timing constraints arise in practice. The transmission deadline $T$ is the latest moment by which the data must be delivered (a QoS constraint set by the application). The service interval $U \geq T$ is the duration before the BS must be active again for the next scheduled transmission. When $U = T$, the BS must be active again the moment the deadline expires; when $U > T$, the BS can finish transmitting early, enter sleep, and transition back to active within the remaining time. The BS can also choose to finish the transmission before the deadline to extend the sleep period.

Consider a BS starting in an active state. It must transmit a data packet of $L$ bits within the time window $T$ and be ready for the next service after some specified time $U \geq T$. The problem of minimizing the total EC becomes:
\begin{equation}
\begin{aligned}
& \underset{P, B, M, t_{\text{tx}}, k}{\text{minimize}}
& & t_{\text{tx}} \left( \frac{P}{\kappa} + \mu + (D_0 + \nu B)M \right) + (U - t_{\text{tx}}) P_{s,k} \\
& \text{subject to}
& & t_{\text{tx}} B \log_2 \left( 1 + \frac{MP\beta}{BN_0} \right) = L, \\
& & & 0 < t_{\text{tx}} \le \min(T, U - \tau_k), \quad 0 < P \le P_{\text{max}}, \\
& & & 0 < B \le B_{\text{max}}, \quad M\in \{1,\ldots ,M_\mathrm{max} \}, \\
& & & k \in \{\text{absolute, deep, idle}\}.
\end{aligned}
\label{eq:latency_optimization}
\end{equation}

The transition time and transmission deadline restrict the maximum transmit time to $t_{\text{tx}} \le \min(T, U - \tau_k)$. Normalizing by $U$, the required average rate is $R = L/U$, and the activity factor is $\alpha = t_{\text{tx}}/U$. The upper bound on the viable activity factor for mode $k$ becomes
\begin{equation}
\alpha_{\mathrm{max},k} \triangleq \min\left(\frac{T}{U}, 1 - \frac{\tau_k}{U}\right).
\end{equation}
This bound reflects two competing physical deadlines: the BS must finish transmitting by time $T$, and it must leave $\tau_k$ seconds to wake up before the time $U$. The BS is constrained by whichever bottleneck occurs first, and actively reaches this upper bound when the latency constraints force the BS to transmit at the higher rate $R/\alpha_{\mathrm{max},k}$ rather than its optimal energy-efficient rate $R_{\rm opt}$.
This allows \eqref{eq:latency_optimization} to be written in a structure similar to the rate-constrained problem in \eqref{eq:1337}---the only difference being the upper bound on $\alpha$:
\begin{equation}
\begin{aligned}
& \underset{P, B, M, \alpha,k}{\text{minimize}}
& & \alpha {\left( \frac{P}{\kappa} + \mu + (D_0 + \nu B)M \right)} + (1 - \alpha) P_{s,k} \\
& \text{subject to}
& & \alpha B \log_2 \left( 1 + \frac{MP\beta}{BN_0} \right) = R, \\
& & & 0 < \alpha \le \alpha_{\text{max},k}, \quad 0 < P \le P_{\text{max}}, \\
& & & 0 < B \le B_{\text{max}}, \quad M\in \{1,\ldots ,M_\mathrm{max}\}, \\
& & & k \in \{\text{absolute, deep, idle}\}.
\end{aligned}
\label{eq:unified_optimization_final}
\end{equation}
Substituting $\alpha$ into the objective function and rewriting, \eqref{eq:unified_optimization_final} becomes equivalent to the EE-maximization problem \eqref{eq:long2} shown at the top of the next page.

\begin{figure*}[!t]
\begin{equation} \label{eq:long2}
\begin{aligned}
\underset{P,B,M,k}{\textrm{maximize}} \quad & \frac{B \log_2\left(1+\frac{MP\beta}{BN_0}\right)}{\frac{P}{\kappa} + \mu + (D_0 + \nu B)M - P_{s,k} + \frac{P_{s,k}}{R} B\log_2\left(1+\frac{MP\beta}{B N_0}\right)} \\
\textrm{subject to} \quad & 0 < P \leq P_\mathrm{max}, \, 0 < B \leq B_\mathrm{max}, \,
M \in \{1,\dots, M_\mathrm{max}\}, \\
& B\log_2\left(1+\frac{MP\beta}{B N_0}\right) \geq \frac{R}{\alpha_{\mathrm{max},k}},
\quad k \in \{\text{absolute, deep, idle}\}.
\end{aligned}
\end{equation}
\hrulefill
\end{figure*}

We propose Algorithm 2 to solve this problem by evaluating the optimal hardware configuration for each sleep mode $k$ in turn. A mode is viable only if $\alpha_{\mathrm{max},k} > 0$. For each viable mode, Algorithm 2 first checks whether the activity factor required to operate at the energy-efficient rate $R_\mathrm{opt}$ from Algorithm 1 fits within $\alpha_{\mathrm{max},k}$. If it does, the transition delay is absorbed into the idle period and the BS uses the unconstrained optimal configuration. If not, the BS is forced to transmit at the higher sustained rate $R / \alpha_{\mathrm{max},k}$ during the entire available active window, with parameters obtained as in Section \ref{sec:highrate}. The mode achieving the highest EE is selected. The cases exhaust the per-sleep-mode optimum, and the comparison across modes makes Algorithm 2 globally EE-optimal \eqref{eq:unified_optimization_final}.

\begin{algorithm}[t]
\label{alg:sleep_selection}
\caption{Sleep Mode and Hardware Configuration under Latency Constraints}
\begin{algorithmic}[1]
\renewcommand{\algorithmicrequire}{\textbf{Input:}}
\renewcommand{\algorithmicensure}{\textbf{Output:}}
\REQUIRE $L, T, U$, sleep modes $k \in \mathcal{K}$, transition times $\tau_k$, sleep powers $P_{s,k}$, $\beta, N_0,\mu, \nu,  D_0, \kappa, P_\mathrm{max}, B_\mathrm{max}, M_\mathrm{max}$
\ENSURE  $k^\star, P^\star, B^\star, M^\star, \mathrm{EE}^\star$
\\ \textit{Initialization}: $\mathrm{EE}^\star \gets 0$, $R = L/U$

\FOR {each sleep mode $k \in \mathcal{K}$}
    \STATE \textit{check viability}: calculate $\alpha_{\mathrm{max},k} = \min\left(\frac{T}{U}, 1 - \frac{\tau_k}{U}\right)$
    \IF {$\alpha_{\mathrm{max},k} \le 0$}
        \STATE \textbf{continue}
    \ENDIF
    
    \STATE \textit{evaluate standard operation}: run Algorithm 1 to find optimal configuration $(P_\mathrm{opt}, B_\mathrm{opt}, M_\mathrm{opt})$, $R_\mathrm{opt}$,  $P_a$
    \STATE Calculate $\alpha = R / R_\mathrm{opt}$
    
    \IF {$\alpha \le \alpha_{\mathrm{max},k}$}
        \STATE Calculate $\mathrm{EE}_k = \frac{R}{\alpha P_a + (1 - \alpha) P_{s,k}}$
        \STATE  $(P_k, B_k, M_k) \gets (P_\mathrm{opt}, B_\mathrm{opt}, M_\mathrm{opt})$
    \ELSE
        \STATE calculate $(P_\mathrm{sust}, B_\mathrm{sust}, M_\mathrm{sust})$ and $P_{a}$ for $R_\mathrm{sust} = R / \alpha_{\mathrm{max},k}$ as in Sec.~\ref{sec:highrate}
        \STATE Calculate $\mathrm{EE}_k = \frac{R}{\alpha_{\mathrm{max},k} P_{a} + (1 - \alpha_{\mathrm{max},k}) P_{s,k}}$
        \STATE $(P_k, B_k, M_k) \gets (P_\mathrm{sust}, B_\mathrm{sust}, M_\mathrm{sust})$
    \ENDIF
    
    \IF {$\mathrm{EE}_k > \mathrm{EE}^\star$}
        \STATE \textit{update global maximum}: $\mathrm{EE}^\star \gets \mathrm{EE}_k$
        \STATE $(k^\star, P^\star, B^\star, M^\star) \gets (k, P_k, B_k, M_k)$
    \ENDIF
\ENDFOR

\RETURN $k^\star, P^\star, B^\star, M^\star, \mathrm{EE}^\star$
\end{algorithmic} 
\end{algorithm}

Fig.~\ref{fig:heatmap} illustrates the optimal sleep mode and EE as a function of the transmission window $T$ from $10$\,ms up to $220$\,ms and payload sizes $L$ from $0.5$\,Gbits to $10.5$\,Gbits. In the figure, we let $T=U$  i.e., the next service interval begins when the transmission deadline expires (relaxing to $U > T$ trivially favors absolute sleep across nearly all $(L, T)$ combinations). The color gradient represents the achieved EE; with lighter regions indicating higher EE. The annotations denote the optimal sleep mode: absolute (``a''), deep (``d''), and idle mode (``i''). The letter (``s'') indicates that the BS does not sleep---but transmits during the full time window at a sustained rate higher than $R_\mathrm{opt}$. Black regions indicate that the payload cannot be delivered within the specified time window given the hardware constraints, i.e., the optimization problem is infeasible. As is evident, increasing $T$ and decreasing $L$ favor absolute and deep sleep modes. Notably, the EE in Fig.~\ref{fig:heatmap} is approximately constant along lines of constant rate $R = L/T$, indicating that the rate requirement is the primary determinant of both the achievable EE and the optimal sleep mode.

\begin{figure}[t!]
	\centering \vspace{-0mm}
	\includegraphics[width=\columnwidth]{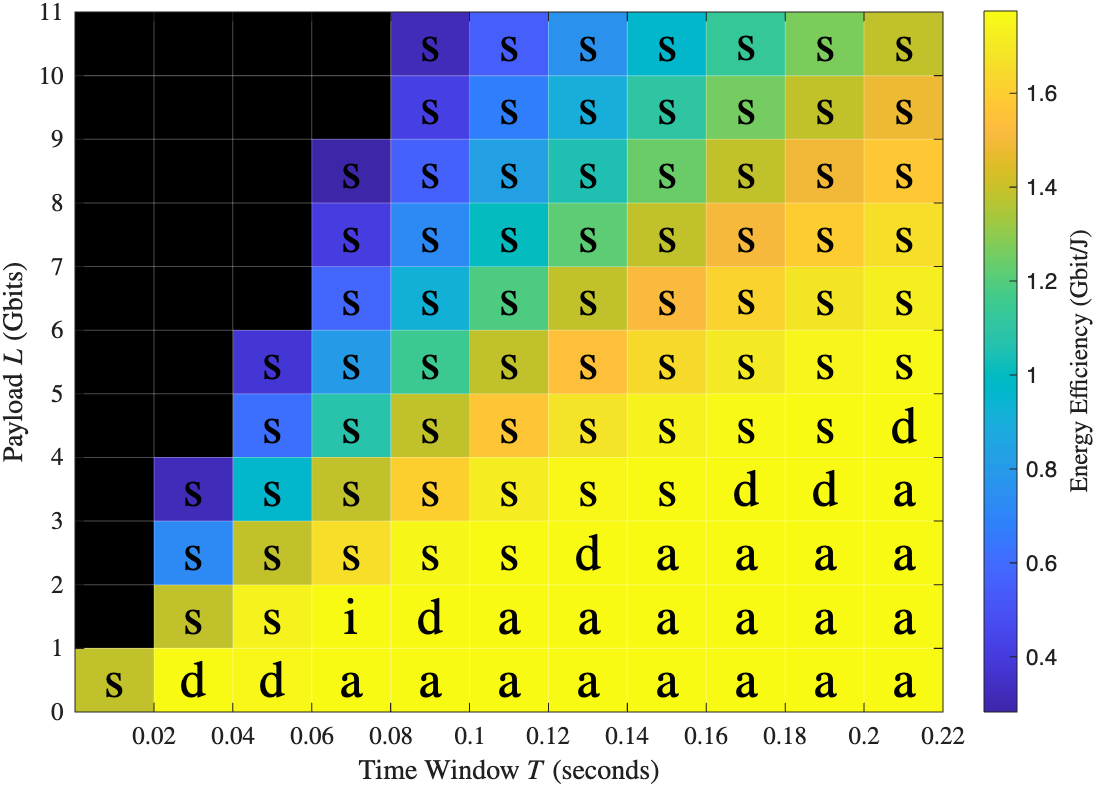}
	\caption{The EE and sleep mode selection for a BS for varying payload sizes and transmission deadlines.}
	\label{fig:heatmap} 
\end{figure}

Fig.~\ref{fig:U_sweep} plots the EE under each sleep mode as the service interval $U$ varies from $T$ to $2T$, for $L = 1.5$\,Gbits and $T = 70$\,ms, corresponding to the ``i'' square in the lower-left region of Fig.~\ref{fig:heatmap}. At $U = T$, idle mode wins: absolute and deep sleep are forced to consume part of the deadline on their wake-up transitions, leaving them with shorter active windows and higher required transmit rates. As $U$ grows, the extra idle time amortizes the wake-up cost. Deep sleep overtakes idle mode shortly after $U = T$, and absolute sleep takes over the optimum near $U = 95$\,ms.

\begin{figure}[t!]
	\centering 
	\includegraphics[width=0.85\columnwidth]{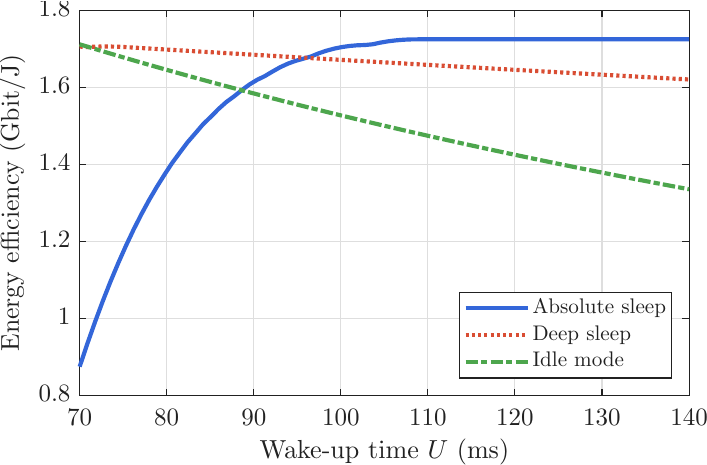}
	\caption{EE achievable under each sleep mode as a function of the service interval $U$.}%
	\label{fig:U_sweep} \vspace{-2mm}
\end{figure}

\section{Conclusion} \label{sec:conclusion}

This paper has explored the fundamentals of EE optimization for wireless communication links and uncovered several structural insights into how transmit power, bandwidth, number of antennas, and sleep-mode depth interact at the EE-optimal operating point. A central finding is that the input transmit power $P/\kappa$ exactly equals the total transceiver power $(D_0 + \nu B)M$ at the optimum, provided neither variable is at its practical upper bound. This equality serves as a design rule for hardware dimensioning. We further developed an algorithm that rapidly converges to the joint EE-maximizing configuration of transmit power, bandwidth, and antennas. 

A particularly striking result is that the EE-optimal SNR collapses to a universal numerical constant of approximately $5.93$\,dB, independent of channel and hardware parameters. This operating point arises as the closed-form optimum of the unconstrained EE objective where $\mu/B \to 0$, and reappears as the solution under deep sleep, where the fixed circuit power $\mu$ becomes a sunk cost paid both across active and sleep periods. The corresponding spectral efficiency maps onto 16-QAM in 3GPP NR \cite{3gpp_ts_38_214} once a practical FEC gap is included. The result provides a deployment-independent target: a macro cell with deep pathloss and a small cell with strong line-of-sight should both aim for the same operating SNR, dimensioning $P$ and $M$ via Theorem \ref{PoverM} to reach it.%

We further extended the framework to incorporate QoS requirements alongside advanced sleep modes, revealing how transition delays and deadline tightness jointly shape the optimal sleep-mode selection. The wider implication is that energy-efficient operation is fundamentally a multi-domain co-design problem: as traffic shifts toward bursty, latency-diverse sessions, treating transmission and sleep as a single coupled optimization rather than two decoupled subsystems is essential for lowering the EC while meeting QoS requirements.

We have deliberately focused on a single-link, narrowband setting to expose the underlying scaling laws. Natural extensions include carrier aggregation, multi-band systems, and wideband channel models. %

\bibliographystyle{IEEEtran}

\bibliography{IEEEabrv,Referenser}

\end{document}